\documentclass{article}
\usepackage{spconf,amsmath,graphicx}
\usepackage{amsfonts}
\usepackage{amsthm}   
\newtheorem{proposition}{Proposition}
\newtheorem{theorem}{Theorem}
\newtheorem{lemma}{Lemma}
\usepackage{cite}

\newtheorem{assumption}{Assumption}
\newenvironment{sketchproof}%
  {\par\noindent{\itshape Sketch of proof.}\hspace*{0.5em}\ignorespaces}%
  {\qed\par}

\title{Error Feedback Approach for Quantization Noise Reduction of Distributed Graph Filters}
\name{Xue Xian Zheng, Tareq Al-Naffouri}
\address{King Abdullah University Of Science And Technology, Saudi Arabia}
\begin{document}
\ninept
\maketitle
\begin{abstract}
This work introduces an error feedback approach for reducing quantization noise of distributed graph filters. It comes from error spectrum shaping techniques from state-space digital filters, and therefore establishes connections between quantized filtering processes over different domains. Quantization noise expression incorporating error feedback for finite impulse response (FIR) and autoregressive moving average (ARMA) graph filters are both derived with regard to time-invariant and time-varying graph topologies. Theoretical analysis is provided, and closed-form error weight coefficients are found. Numerical experiments demonstrate the effectiveness of the proposed method in noise reduction for the graph filters regardless of the deterministic and random graph topologies. 
\end{abstract}

\section{Introduction}
\label{sec:intro}

Graph filters are the fundamental blocks for processing signals that reside on graphs \cite{1,2,3}. Recent studies have investigated the quantization problem for their distributed implementations, where nodes are assumed to have finite bits of communication capability. In this setting, the information loss caused by quantized message exchanges will inevitably degrade the filtering performance.

To reduce this effect, several techniques have been proposed \cite{4,5,6,7,8}. These methods either focus on improving filter robustness or optimizing quantization strategies. This paper provides a different perspective to deal with quantization effects. We introduce the error feedback as post-mitigation to the quantization errors. Our method extends error spectrum shaping (ESS) techniques, commonly used in digital filters \cite{9,10,11,12,13}, to distributed graph filters. In our approach, the error feedback matrix is constrained to be diagonal, equivalently meaning each node stores its quantization error locally, applies a weight factor, and participates in the fusion process as a compensation mechanism. We derive the quantization noise expressions incorporating error feedback for both finite impulse response (FIR) \cite{14} and autoregressive moving average (ARMA) \cite{15} graph filters over time-invariant and time-varying graph topologies. Additionally, we provide theoretical analyses and closed-form solutions for the error weight coefficients.

Our paper is organized as follows: In Section II, we briefly review the quantization problem for distributed graph filters. In Section III, we present the proposed error feedback method. Section IV provides experimental results to demonstrate the effectiveness. The conclusion is drawn in Section V.

\textit{Notations:} In this paper, for a matrix $\mathbf{X}$, $\mathbf{X}^{T}$ denotes its transpose, and $[\mathbf{X}]_{ij}$ represents the element in the $i$-th row and $j$-th column. The $i$-th row vector is denoted by $[\mathbf{X}]_{i:}$, while $[\mathbf{X}]_{:j}$ represents the $j$-th column vector. The trace of the matrix $\mathbf{X}$ is denoted by $\text{tr}(\mathbf{X})$, and $\Vert \mathbf{X} \Vert_{2}$ indicates its spectral norm. We use $\circ$ for the Hadamard product, and $\prod^{\leftarrow}$ and $\prod^{\rightarrow}$ for left and right continued matrix multiplication, respectively.

\section{Problem Formulation}
\label{sec:format}

Consider a signal $\mathbf{x} = [x_1, x_2, \dots, x_N]^{T} \in \mathbb{R}^{N}$ that resides on a graph $\mathcal{G} = (\mathcal{V}, \mathcal{E})$, where $\mathcal{V} = \{1, \dots , N\}$ represents the set of nodes, and $\mathcal{E} \subseteq \mathcal{V} \times \mathcal{V}$ represents the set of edges, with $|\mathcal{E}| = M$. The signal $\mathbf{x}$ is defined on the nodes of the graph, i.e., $\mathbf{x}: \mathcal{V} \to \mathbb{R}^{N}$. We denote the edge $(j, i) \in \mathcal{E}$ as a connection from node $j$ to node $i$. The set of all nodes connected to node $i$ is denoted by $\mathcal{N}_i = \{j \mid (j, i) \in \mathcal{E}\}$. The filtering process on $\mathbf{x}$ can be taken into the following scenarios:
\vspace{-2mm}
\subsection{Deterministic graph filters}

Let the node connections be stable, resulting in a deterministic graph filtering process. The graph shift operator $\mathbf{S} \in \mathbb{R}^{N \times N}$ encodes graph structure, represented by the adjacency matrix $\mathbf{A}$, the Laplacian $\mathbf{L} = \text{diag}(\mathbf{A}\mathbf{1}) - \mathbf{A}$ or any of their normalized or translated variants. For simplicity, we assume $\mathbf{S} = \mathbf{S}^{T}$ (i.e., $\mathbf{S}$ is symmetric) with a bounded spectral radius $\rho$. This allows us to express $\mathbf{S} = \mathbf{U} \Lambda \mathbf{U}^{T}$, where $\Lambda = \text{diag}\{\lambda_{1}, \dots, \lambda_{N}\}$ is a diagonal matrix of real eigenvalues. The following filters are then applied:

\noindent \textbf{Deterministic FIR:} $\mathbf{y} = \mathbf{H} (\mathbf{S}) \mathbf{x} = \sum_{k=0}^K \phi_k \mathbf{S}^k \mathbf{x}$ with an equivalent distributed state-space realization as
 \begin{equation}
\begin{aligned}
    \mathbf{w}_{k} = \mathbf{S}\mathbf{w}_{k-1} ~ \text{for } \mathbf{w}_{0}=\mathbf{x}, \mathbf{y} =   \sum_{k=0}^{K}\phi_{k}\mathbf{w}_{k}
\end{aligned}
\label{deqn_ex1}
\end{equation}

\noindent \textbf{Deterministic ARMA:} Under the conditions $|\psi_k \rho| < 1$, $\mathbf{y}_{t}$ inherently follows a distributed state-space realization
  \begin{equation}
\begin{aligned}
\mathbf{w}^{(k)}_{t} &= \psi_k \mathbf{S}\mathbf{w}^{(k)}_{t-1} + \varphi_k\mathbf{x} \\
\mathbf{y}_{t} &= \sum_{k=1}^{K}\mathbf{w}^{(k)}_{t} \quad \text{for } t \geq 1
\end{aligned}
\label{deqn_ex2}
\end{equation}
\vspace{-4mm}
\subsection{Stochastic graph filters}

When node connections may be subject to random failures, it leads to a stochastic graph filtering process. More concretely,

\vspace{1em}
\noindent \textbf{Definition 1} (RES model \cite{16})\textbf{. }\textit{Let $\mathcal{G} = (\mathcal{V}, \mathcal{E})$ be the underlying graph with edge $(j,i) \in \mathcal{E}$ that may fail. Its realization $\mathcal{G}_t = (\mathcal{V}, \mathcal{E}_t)$ is formed with edge sampling $\text{Pr}[(i,j)\in \mathcal{E}_{t}] = p, 0 \leq p < 1 $ for all $(i,j)\in \mathcal{E}$ independently.}
\vspace{1em}

\noindent \textit{Remark:} Under \textbf{Definition 1}, the graph shift operator of RES model is random, generated using a random Bernoulli matrix $\mathbf{B}_{t}$ as a mask \cite{17}: $\mathbf{S}_{t} = \mathbf{A}_{t} = \mathbf{B}_{t} \circ \mathbf{A}$ or $\mathbf{S}_{t} = \mathbf{L}_{t} = \text{diag}(\mathbf{A}_{t}\mathbf{1}) - \mathbf{A}_{t}$, with $\overline{\mathbf{S}} = \mathbb{E}[\mathbf{S}_{t}] = p \mathbf{A}$ or $\overline{\mathbf{S}} = \mathbb{E}[\mathbf{S}_{t}] = \mathbb{E}[\text{diag}(\mathbf{A}_{t}\mathbf{1})]- p \mathbf{A}$. Beyond these expressions, $\Vert \mathbf{S}_{t} \Vert_{2} \leq \rho $ holds, we can then apply filters as

\noindent \textbf{Stochastic FIR:} $\mathbf{y}= \mathbf{H} (\mathcal{\Phi}_{0:K-1} )\mathbf{x} =\sum_{k=0}^{K} \phi_k \mathcal{\Phi}_{0:k-1} \mathbf{x}$ with $\mathcal{\Phi}_{t:t'} = \prod_{\tau = t}^{t'\leftarrow}\mathbf{S}_{\tau}$ if $t'\geq t$ and $\mathbf{I}$ if $t'<t$. Then 
 \begin{equation}
\begin{aligned}
    \mathbf{w}_{k} = \mathbf{S}_{k-1}\mathbf{w}_{k-1} ~ \text{for } \mathbf{w}_{0}=\mathbf{x}, \mathbf{y} =   \sum_{k=0}^{K}\phi_{k}\mathbf{w}_{k}
\end{aligned}
\label{deqn_ex3}
\end{equation}
\noindent \textbf{Stochastic ARMA:} With $\Vert \psi_k \mathbf{S}_{t} \Vert_{2} \leq |\psi_k\rho|< 1$,
  \begin{equation}
\begin{aligned}
\mathbf{w}^{(k)}_{t} &= \psi_k \mathbf{S}_{t-1}\mathbf{w}^{(k)}_{t-1} + \varphi_k\mathbf{x} \\
\mathbf{y}_{t} &= \sum_{k=1}^{K}\mathbf{w}^{(k)}_{t} \quad \text{for } t \geq 1
\end{aligned}
\label{deqn_ex4}
\end{equation}
\vspace{-5mm}
\subsection{Quantization model}

The realizations in Eqs.~(\ref{deqn_ex1}--\ref{deqn_ex4}) primarily differ from standard state-space models in their state transition matrix, $\mathbf{S}$ or $\mathbf{S}_t$, which is structured by the underlying graph topology. This matrix inherently encapsulates the way of information exchange between state variables. The quantization problem arises due to the finite-bit constraint on information exchange. The formalization of the problem is below:
\begin{assumption}
    Let the state variables be $\mathbf{w}$. The fusion process $\mathbf{S}\mathbf{w}$ or $\mathbf{S}_{t}\mathbf{w}$ occurs after $\mathcal{Q}[\mathbf{w}] = \mathbf{w} + \mathbf{n}$, where $\mathcal{Q}[\cdot]$ is a uniform quantizer over $[-r, r]$ with $b$-bit signed representation, yielding $\mathbb{E}[\mathbf{n}] = \mathbf{0}$ and $\Sigma_{\mathbf{n}} = \sigma^2 \mathbf{I} = r / (3 \cdot 2^{b})$.
\end{assumption}
\noindent \textit{Remark:} Under \textbf{Assumption 1}, we obtain the following noise propagation models for Eqs.~(\ref{deqn_ex1}--\ref{deqn_ex4}):
\setlength{\abovedisplayskip}{3pt}
\setlength{\belowdisplayskip}{3pt}
\begin{subequations}
\begin{align}
    \widetilde{\mathbf{w}}_{k} &= \mathbf{S}(\widetilde{\mathbf{w}}_{k-1}+\mathbf{n}_{k-1}), \widetilde{\mathbf{y}} = \sum_{k=0}^{K}\phi_{k}\widetilde{\mathbf{w}}_{k} \tag{5a} \label{deqn_ex5} \\
    \widetilde{\mathbf{w}}^{(k)}_{t} &= \psi_k \mathbf{S}(\widetilde{\mathbf{w}}^{(k)}_{t-1}+\mathbf{n}^{(k)}_{t-1}), \widetilde{\mathbf{y}}_{t} = \sum_{k=1}^{K}\widetilde{\mathbf{w}}^{(k)}_{t} \tag{5b} \label{deqn_ex6} \\
    \widetilde{\mathbf{w}}_{k} &= \mathbf{S}_{k-1}(\widetilde{\mathbf{w}}_{k-1}+\mathbf{n}_{k-1}),  \widetilde{\mathbf{y}} = \sum_{k=0}^{K}\phi_{k}\widetilde{\mathbf{w}}_{k} \tag{5c} \label{deqn_ex7} \\
    \widetilde{\mathbf{w}}^{(k)}_{t} &= \psi_k \mathbf{S}_{t-1}(\widetilde{\mathbf{w}}^{(k)}_{t-1}+\mathbf{n}^{(k)}_{t-1}), \widetilde{\mathbf{y}}_{t} = \sum_{k=1}^{K}\widetilde{\mathbf{w}}^{(k)}_{t} \tag{5d} \label{deqn_ex8}
\end{align}
\end{subequations} 
where $\widetilde{\mathbf{w}}_{0}=\widetilde{\mathbf{w}}^{(k)}_{0}=\mathbf{0}$, $\mathbf{n}_{k-1} \sim (\mathbf{0}, \sigma^{2}_{k-1}\mathbf{I})$ i.i.d. uniform distribution for $k$, $\mathbf{n}^{(k)}_{t-1} \sim (\mathbf{0}, (\sigma^{(k)}_{t-1})^{2}\mathbf{I})$ i.i.d. uniform distribution for $k$ and $t$, respectively. We simplify $\sigma^{(k)}_{t-1} \to \sigma^{(k)}$ for the subsequent section discussions. 

\begin{figure}
  \centering
    \setlength{\abovecaptionskip}{0pt}  
  \setlength{\belowcaptionskip}{0pt}  
  \includegraphics[width=0.48\textwidth]{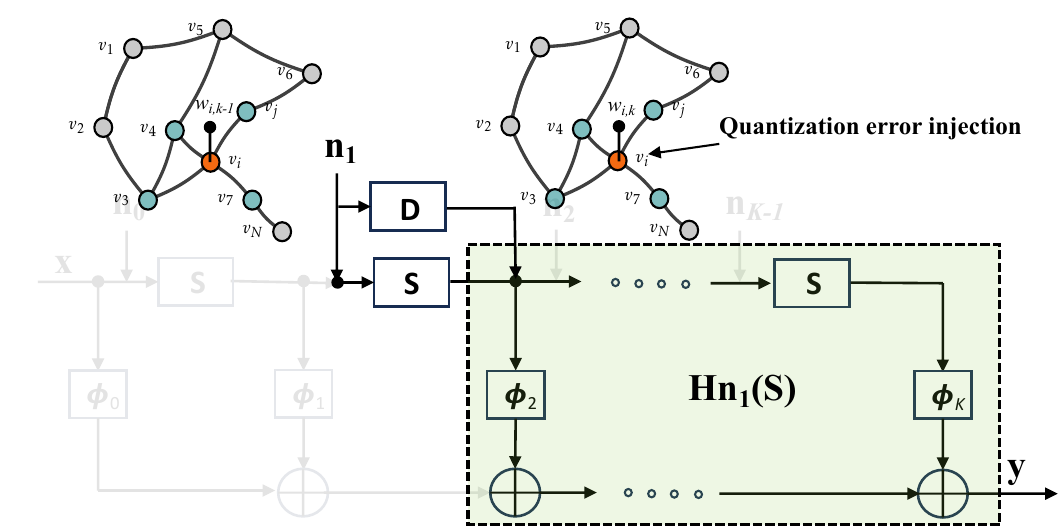}\\
  \caption{An illustration of quantization error feedback on deterministic FIR graph filter. Let $\mathbf{D}$ be diagonal, it is interpreted that each node locally stores its quantization error, applies a weight factor, and then integrates it into the fusion process.}
  \label{fig:diagram}
\end{figure}

\section{Quantization Error Feedback}
In this section, we present the method of quantization error feedback to mitigate noise propagation in Eqs.~(\ref{deqn_ex5}--\ref{deqn_ex8}). Specifically, we derive the feedbacked quantization noise expressions and closed-form feedback weight coefficients. The main idea of our method is illustrated in Fig.~\ref{fig:diagram}. Instead of direct injection, each quantization error is additionally weighted by a diagonal matrix $\mathbf{D}$ and integrated into subsequent processing as a compensation mechanism. 
\label{sec:pagestyle}
\vspace{-4mm}
\subsection{Error feedback on deterministic FIR graph filter}

\label{subsec:FIR}
In this case, we evaluate the power of the output noise from each source \(\mathbf{n}_{k-1}\). Given that these noise sources are independent, the total noise power can be expressed as the sum of their individual contributions. Our findings are summarized as:

\begin{proposition}
For feedbacked quantization noise output $\widetilde{\mathbf{y}}_{\mathbf{n}_{k-1}} = \mathbf{H}_{\mathbf{n}_{k-1}} (\mathbf{S})(\mathbf{S}-\mathbf{D})\mathbf{n}_{k-1}$ coming from $\mathbf{n}_{k-1}$ with sub graph filter $\mathbf{H}_{\mathbf{n}_{k-1}} (\mathbf{S}) =  \sum_{\kappa =k}^{K}\phi_{\kappa}\mathbf{S}^{\kappa - k}$, its average noise power per node $ \zeta_{k-1} =\frac{1}{N} \textnormal{tr}\Big( \mathbb{E}[\widetilde{\mathbf{y}}_{\mathbf{n}_{k-1}} \widetilde{\mathbf{y}}^{T}_{\mathbf{n}_{k-1}}]\Big) $ satisfies
\begin{equation}
\begin{aligned}
    \zeta_{k-1} & = \frac{\sigma_{k-1}^2}{N} \bigg\{\underbrace{\textnormal{tr}\Big(\mathbf{G}_{\mathbf{n}_{k-1}}(\mathbf{S}) \mathbf{S}^{2}\Big)}_{\text{Original noise gain}} \\
    & \underbrace{+ \textnormal{tr}\Big(\mathbf{G}_{\mathbf{n}_{k-1}}(\mathbf{S}) \mathbf{D}^{2}\Big)-2\textnormal{tr}\Big(\mathbf{G}_{\mathbf{n}_{k-1}}(\mathbf{S}) \mathbf{S}\mathbf{D}\Big)}_{\text{Mitigation:} I_{k-1}<0} \bigg\}
\end{aligned}
\label{deqn_ex9}
\end{equation}
where $\mathbf{G}_{\mathbf{n}_{k-1}}(\mathbf{S})=\mathbf{H}^{T}_{\mathbf{n}_{k-1}} (\mathbf{S})\mathbf{H}_{\mathbf{n}_{k-1}} (\mathbf{S})$ is the Gram matrix.
\end{proposition}
\begin{sketchproof}
  By $\mathbb{E}[\mathbf{n}_{k-1}\mathbf{n}^{T}_{k-1}]=\sigma_{k-1}^2\mathbf{I}$, $\mathbb{E}[\widetilde{\mathbf{y}}_{\mathbf{n}_{k-1}} \widetilde{\mathbf{y}}^{T}_{\mathbf{n}_{k-1}}]=\sigma_{k-1}^2(\mathbf{H}_{\mathbf{n}_{k-1}}(\mathbf{S})(\mathbf{S}-\mathbf{D})(\mathbf{S}-\mathbf{D})^{T}\mathbf{H}^{T}_{\mathbf{n}_{k-1}} (\mathbf{S}))$ . By trace cyclic property \(\text{tr}(\mathbf{AB}) = \textnormal{tr}(\mathbf{BA}),\mathbf{S}=\mathbf{S}^{T}\),  Eq.~(6) follows.
\end{sketchproof}
\begin{theorem}
Under \textbf{Proposition 1}, given $\mathbf{D}$ varies as $\mathbf{D}_{k-1}=\textnormal{diag}\{\alpha_{1, k-1}, \dots, \alpha_{N, k-1}\}$, $\Theta = \{\mathbf{D}_{0}\mathbf{1}, \dots, \mathbf{D}_{K-1}\mathbf{1}\} \in \mathbb{R}^{N \times K}$, the closed-form solution $\alpha_{i, k-1}$ and noise reduction $I(\Theta)$ are
\begin{equation}
\begin{aligned}
\alpha_{i,k-1} & = \frac{\Big [\mathbf{G}_{\mathbf{n}_{k-1}}(\mathbf{S})\mathbf{S}\Big ]_{ii}}{\Big [ \mathbf{G}_{\mathbf{n}_{k-1}}(\mathbf{S})\Big ]_{ii}}\\
I(\Theta) & =\sum_{k=1}^{K}\frac{\sigma_{k-1}^2}{N}\sum_{i=1}^{N}{\frac{\Big [\mathbf{G}_{\mathbf{n}_{k-1}}(\mathbf{S})\mathbf{S}\Big ]^{2}_{ii}}{\Big [ \mathbf{G}_{\mathbf{n}_{k-1}}(\mathbf{S})\Big ]_{ii}}}
\end{aligned}
\label{deqn_ex10}
\end{equation}
\end{theorem}
\begin{proof}
See Appendix I.
\end{proof}

\noindent \textit{Remark:} \textbf{Theorem 1} can be extended to yield closed-form solutions in partially invariant cases as $\alpha_{i,k-1} \to \alpha_{k-1}$ with $\mathbf{D}_{k-1} = \alpha_{k-1}\mathbf{I}, \Theta \in \mathbb{R}^{K}$ or $\alpha_{i,k-1} \to \alpha_{i}$ with $\mathbf{D}_{k-1} = \mathbf{D}, \Theta \in \mathbb{R}^{N}$, thereby reducing the storage and computations.
\vspace{-3mm}
\subsection{Error feedback on deterministic ARMA graph filter}
\vspace{-1mm}
In this case, we evaluate the asymptotic noise power of each $k$ parallel $\widetilde{\mathbf{w}}^{(k)}_{t\to \infty}$, which is summarized as follows:
\begin{proposition}
    For feedbacked quantization noise model $\widetilde{\mathbf{w}}^{(k)}_{t} = \psi_k \mathbf{S}\widetilde{\mathbf{w}}^{(k)}_{t-1}+(\psi_k \mathbf{S}-\mathbf{D})\mathbf{n}^{(k)}_{t-1}$, the asymptotic average noise power per node $ \zeta^{(k)}_{t \to \infty} =\frac{1}{N} \textnormal{tr}\Big( \mathbb{E} [\widetilde{\mathbf{w}}^{(k)}_{t\to \infty} (\widetilde{\mathbf{w}}^{(k)}_{t\to \infty})^{T} ]\Big)$ satisfies
  \begin{equation}
\begin{aligned}
 \zeta^{(k)}_{t \to \infty} & =\frac{(\sigma^{(k)})^2}{N}\textnormal{tr} \Big( (\psi_{k}\mathbf{S}-\mathbf{D})^{T}\mathbf{W}^{(k)}_0(\psi_{k}\mathbf{S}-\mathbf{D}) \Big) =\frac{(\sigma^{(k)})^2}{N}\Big\{ \\ 
 &
\underbrace{\textnormal{tr}(\psi_{k}^{2}\mathbf{W}^{(k)}_0 \mathbf{S}^{2}) }_{\textit{Original noise gain}}\underbrace{+ \textnormal{tr}(\mathbf{W}^{(k)}_0\mathbf{D}^{2})-2\textnormal{tr}(\psi_{k}\mathbf{W}^{(k)}_0\mathbf{S}\mathbf{D})}_{\textit{Mitigation:} I^{(k)}<0} \Big\} 
\end{aligned}
\end{equation}
where $\mathbf{W}^{(k)}_0$ is the observability Gramian of $k$ branch dynamics, satisfying Lyapunov equation $\mathbf{W}^{(k)}_0 = \psi_{k}^{2}\mathbf{S}\mathbf{W}^{(k)}_0 \mathbf{S}+\mathbf{I}$
\end{proposition}
\begin{sketchproof}
     $\widetilde{\mathbf{w}}^{(k)}_{t} = \sum_{\tau=0}^{t-1}(\psi_{k}\mathbf{S})^{t-\tau-1}(\psi_{k}\mathbf{S}-\mathbf{D})\mathbf{n}^{(k)}_{\tau}$. By $\mathbb{E}[\mathbf{n}^{(k)}_{\tau}(\mathbf{n}^{(k)}_{\tau})^{T}] = (\sigma^{(k)})^2 \mathbf{I}$, \(\text{tr}(\mathbf{AB}) = \textnormal{tr}(\mathbf{BA})\), $\mathbb{E}[\widetilde{\mathbf{w}}^{(k)}_{t}(\widetilde{\mathbf{w}}^{(k)}_{t})^{T}]=(\sigma^{(k)})^2\sum_{\tau=0}^{t-1}(\psi_{k}\mathbf{S}-\mathbf{D})^{T}(\psi_{k}\mathbf{S}^{T})^{t-\tau-1}(\psi_{k}\mathbf{S})^{t-\tau-1}(\psi_{k}\mathbf{S}-\mathbf{D})$. As $t\to\infty$, $\mathbf{W}^{(k)}_0=\sum_{t=0}^{\infty}(\psi_{k}\mathbf{S})^{t}(\psi_{k}\mathbf{S})^{t}$. Eq. (8) follows.
\end{sketchproof}
\begin{theorem}
    Under \textbf{Proposition 2}, given $\mathbf{D}$ varies as $\mathbf{D}^{(k)}=\textnormal{diag}\{\alpha_{1}^{(k)},\cdots, \alpha_{N}^{(k)} \}$ and $\Theta = \{\mathbf{D}^{(1)}\mathbf{1},\cdots, \mathbf{D}^{(K)}\mathbf{1}\}\in\mathbb{R}^{N\times K}$, the closed-form solution  $\alpha_{i}^{(k)}$ and noise reduction $I(\Theta)$ are
\begin{equation}
\begin{aligned}
\alpha_{i}^{(k)} = \frac{\psi_{k}[\mathbf{W}_0^{(k)}\mathbf{S}]_{ii}}{[\mathbf{W}_0^{(k)}]_{ii}},
I(\Theta) =\sum_{k=1}^{K}\frac{(\sigma^{(k)})^2}{N}\sum_{i=1}^{N}\frac{\psi_{k}^{2}[\mathbf{W}_0^{(k)}\mathbf{S}]_{ii}^{2}}{[\mathbf{W}_0^{(k)}]_{ii}}
\end{aligned}
\end{equation}
\end{theorem}
\begin{proof}
    The argument follows similarly to that of \textbf{Theorem 1}.
\end{proof}
\vspace{-4mm}
\subsection{Error feedback on stochastic FIR graph filter}

In this case, we follow the analysis of Subsection \ref{subsec:FIR} as below:
\begin{proposition}
    For feedbacked quantization noise output $\widetilde{\mathbf{y}}_{\mathbf{n}_{k-1}} = \mathbf{H}_{\mathbf{n}_{k-1}} (\mathcal{\Phi}_{k:K-1})(\mathbf{S}_{k-1}-\mathbf{D})\mathbf{n}_{k-1}$ coming from $\mathbf{n}_{k-1}$ with stochastic sub graph filter  $
    \mathbf{H}_{\mathbf{n}_{k-1}} (\mathcal{\Phi}_{k:K-1}) =   \sum_{\kappa =k}^{K}\phi_{\kappa}\mathcal{\Phi}_{k:\kappa-1}$, its average noise power per node $ \zeta_{k-1} =\frac{1}{N} \textnormal{tr}\Big( \mathbb{E}[\widetilde{\mathbf{y}}_{\mathbf{n}_{k-1}} \widetilde{\mathbf{y}}^{T}_{\mathbf{n}_{k-1}}]\Big)$ satisfies
    \begin{equation}
\begin{aligned}
      & \zeta_{k-1} = \frac{\sigma_{k-1}^2}{N} \bigg\{\underbrace{\textnormal{tr}\Big(\mathbb{E}[\mathbf{G}_{\mathbf{n}_{k-1}}(\mathcal{\Phi}_{k:K-1}) \mathbf{S}_{k-1}^{2}]\Big)}_{\textit{Original noise gain}} + \\ & \underbrace{ \textnormal{tr}\Big(\mathbb{E}[\mathbf{G}_{\mathbf{n}_{k-1}}(\mathcal{\Phi}_{k:K-1})] \mathbf{D}^{2}\Big)-2\textnormal{tr}\Big(\mathbb{E}[\mathbf{G}_{\mathbf{n}_{k-1}}(\mathcal{\Phi}_{k:K-1})] \overline{\mathbf{S}}\mathbf{D}\Big) }_{\textit{Mitigation:} I_{k-1}<0}\bigg\}
\end{aligned}
\end{equation}
where $\mathbf{G}_{\mathbf{n}_{k-1}}(\mathcal{\Phi}_{k:K-1})=\mathbf{H}^{T}_{\mathbf{n}_{k-1}} (\mathcal{\Phi}_{k:K-1})\mathbf{H}_{\mathbf{n}_{k-1}} (\mathcal{\Phi}_{k:K-1})$.
\end{proposition}
\begin{sketchproof}
    By the law of total expectation \(\mathbb{E}[\mathbf{X}] = \mathbb{E}[\mathbb{E}[\mathbf{X}|\mathbf{Y}]]\), $\mathbb{E}[\widetilde{\mathbf{y}}_{\mathbf{n}_{k-1}} \widetilde{\mathbf{y}}^{T}_{\mathbf{n}_{k-1}}] = \mathbb{E}[\mathbf{H}_{\mathbf{n}_{k-1}} (\mathcal{\Phi}_{k:K-1})(\mathbf{S}_{k-1}-\mathbf{D})\mathbb{E}[\mathbf{n}_{k-1}\mathbf{n}^{T}_{k-1}]\\(\mathbf{S}_{k-1}-\mathbf{D})^{T}\mathbf{H}^{T}_{\mathbf{n}_{k-1}} (\mathcal{\Phi}_{k:K-1})]$. As \(\mathbb{E}[\mathbf{n}_{k-1}\mathbf{n}^{T}_{k-1}]=\sigma_{k-1}^2\mathbf{I}\) and \(\mathbb{E}[\mathbf{S}_{k-1}]=\overline{\mathbf{S}}\), by using \(\textnormal{tr}(\mathbb{E}[\cdot]) = \mathbb{E}[\textnormal{tr}(\cdot)]\) and \(\text{tr}(\mathbf{AB}) = \textnormal{tr}(\mathbf{BA})\), Eq.~(10) follows.
\end{sketchproof}
\begin{theorem}
    Under \textbf{Proposition 3}, given $\mathbf{D}$ varies as $\mathbf{D}_{k-1}=\textnormal{diag}\{\alpha_{1, k-1},\cdots, \alpha_{N, k-1} \}$, $\Theta =\{\mathbf{D}_{0}\mathbf{1},\cdots,\mathbf{D}_{K-1}\mathbf{1}\}\in\mathbb{R}^{N\times K}$, the closed-form solution $\alpha_{i, k-1}$ and noise reduction $I(\Theta)$ are
    \begin{equation}
\begin{aligned}
\alpha_{i,k-1} & = \frac{\Big [\mathbb{E}[\mathbf{G}_{\mathbf{n}_{k-1}}(\mathcal{\Phi}_{k:K-1})] \overline{\mathbf{S}} \Big ]_{ii}}{\Big [\mathbb{E}[\mathbf{G}_{\mathbf{n}_{k-1}}(\mathcal{\Phi}_{k:K-1})] \Big ]_{ii}}\\
I(\Theta) & =\sum_{k=1}^{K}\frac{\sigma_{k-1}^2}{N}\sum_{i=1}^{N}{\frac{\Big [\mathbb{E}[\mathbf{G}_{\mathbf{n}_{k-1}}(\mathcal{\Phi}_{k:K-1})] \overline{\mathbf{S}}\Big ]^{2}_{ii}}{\Big [ \mathbb{E}[\mathbf{G}_{\mathbf{n}_{k-1}}(\mathcal{\Phi}_{k:K-1})]\Big ]_{ii}}}
\end{aligned}
\end{equation}
\end{theorem}
\begin{proof}
    The argument follows similarly to that of \textbf{Theorem 1}.
\end{proof}
\noindent \textit{Remark:} The bottleneck here is deriving \(\mathbb{E}[\mathbf{G}_{\mathbf{n}_{k-1}}(\mathcal{\Phi}_{k:K-1})]\), which involves a set of product of structured random matrices $\mathcal{\Phi}_{k:K-1}$. Using the linearity of expectation $\mathbb{E}[\mathbf{X}+\mathbf{Y}] = \mathbb{E}[\mathbf{X}]+\mathbb{E}[\mathbf{Y}]$, $\mathbb{E}[\mathbf{G}_{\mathbf{n}_{k-1}}(\mathcal{\Phi}_{k:K-1})]$ is expressed as $\sum_{\kappa_1 =k}^{K}\sum_{\kappa_2 =k}^{K}\phi_{\kappa_1}\phi_{\kappa_2}\\\mathbb{E}[\mathcal{\Phi}_{k:\kappa_1-1}^{T}\mathcal{\Phi}_{k:\kappa_2-1}]$. Let $\kappa_1 < \kappa_2$ and $\mathcal{\Phi}^{T}_{t:t'} = \prod_{\tau = t}^{t'\rightarrow}\mathbf{S}_{\tau}$, we have 
\begin{equation}
\begin{aligned}
\mathbb{E}[\mathcal{\Phi}_{k:\kappa_1-1}^{T}\mathcal{\Phi}_{k:\kappa_2-1}]  = \mathbb{E}\Big[\overset{\xrightarrow{\hspace{1.2cm}}}{\mathbf{S}_{k} \cdot \cdot \mathbf{S}_{\kappa_1-1}}\overset{\xleftarrow{\hspace{2.7cm}}}{
\mathbf{S}_{\kappa_2-1}\cdot \cdot \mathbf{S}_{\kappa_1-1} \cdot \cdot \mathbf{S}_{k}}
\Big]
\end{aligned}
\end{equation}
which is further simplified to $\mathbb{E}[\mathbf{S}_{k} \cdot \cdot \mathbb{E}[\mathbf{S}_{\kappa_1-1}\overline{\mathbf{S}}^{\kappa_2-\kappa_1}\mathbf{S}_{\kappa_1-1}]\cdot \cdot\mathbf{S}_{k}]$ using \(\mathbb{E}[\mathbf{X}] = \mathbb{E}[\mathbb{E}[\mathbf{X}|\mathbf{Y}]]\). By observing that Eq.~(12) implies a recursive computational structure, we design an algorithm to address the bottleneck above, outlined as follows:

\noindent \textbf{Kernel method:} For any $\kappa_1,\kappa_2$, we let $\mathbf{M}_{0}=\overline{\mathbf{S}}^{|\kappa_2-\kappa_1|}$ and recursively calculate $\mathbf{M}_{l}=\mathbb{E}[\mathbf{S}_{t}\mathbf{M}_{l-1}\mathbf{S}_{t}]$ until $l= \textnormal{min} \{\kappa_1,\kappa_2\}-k$, to obtain $\mathbb{E}[\mathcal{\Phi}_{k:\kappa_1-1}^{T}\mathcal{\Phi}_{k:\kappa_2-1}]=\mathbf{M}_{l}$. Note that the recursive nature ensures every $\mathbf{S}_{t}$ involved in the calculation is independent with the previous results, so the subscripts used in Eq.~(12) to distinguish independence like $k,\kappa_1,\kappa_2$ are omitted. Furthermore, by defining the kernel tensor $\textit{\textbf{Ker}}(i,j)= \mathbb{E}\Big[[\mathbf{S}_{t}]_{:j} [\mathbf{S}_{t}]_{i:}\Big]$, we enable parallel computation of the expectation
\begin{equation}
\begin{aligned}
\Big[\mathbb{E}[\mathbf{S}_{t}\mathbf{M}_{l-1}\mathbf{S}_{t}]\Big]_{ij} = \text{tr}\Big(\mathbf{M}_{l-1}  \textit{\textbf{Ker}}(i,j)\Big)
\end{aligned}
\end{equation}
Here \(\textit{\textbf{Ker}}(i,j)\) captures the mutual uncertainty between nodes \(i\) and \(j\) about their connections to the rest of the network. Since this kernel tensor can be precomputed before any iterations, it allows us to accelerate the overall computation by reducing the complexity of the repeated matrix operations.

\vspace{-3mm}
\subsection{Error feedback on stochastic ARMA graph filter}
In this case, we have the following results: 
\begin{proposition}
    For feedbacked quantization noise model $\widetilde{\mathbf{w}}^{(k)}_{t} = \psi_k \mathbf{S}_{t-1}\widetilde{\mathbf{w}}^{(k)}_{t-1}+(\psi_k \mathbf{S}_{t-1}-\mathbf{D})\mathbf{n}^{(k)}_{t-1}$, the asymptotic average noise power per node $ \zeta^{(k)}_{t \to \infty} =\frac{1}{N} \textnormal{tr}\Big( \mathbb{E} [\widetilde{\mathbf{w}}^{(k)}_{t\to \infty} (\widetilde{\mathbf{w}}^{(k)}_{t\to \infty})^{T} ]\Big)$
    \begin{equation}
\begin{aligned}
& =\frac{(\sigma^{(k)})^2}{N}\textnormal{tr} \Big( \mathbb{E}[ (\psi_{k}\mathbf{S}_{t}-\mathbf{D})^{T}\mathbf{W}^{(k)}_{\mathcal{\Phi}}(\psi_{k}\mathbf{S}_{t}-\mathbf{D})]\Big)= \frac{(\sigma^{(k)})^2}{N}\Big\{ \\ & 
\underbrace{\textnormal{tr}(\psi_{k}^{2}\mathbb{E} [\mathbf{W}^{(k)}_{\mathcal{\Phi}} \mathbf{S}_{t}^{2}]) }_{\textit{Original noise gain}}\underbrace{+ \textnormal{tr}(\mathbf{W}^{(k)}_{\mathcal{\Phi}}\mathbf{D}^{2})-2\textnormal{tr}(\psi_{k}\mathbf{W}^{(k)}_{\mathcal{\Phi}}\overline{\mathbf{S}}\mathbf{D})}_{\textit{Mitigation:} I^{(k)}<0} \Big\} 
\end{aligned}
\end{equation}
where $\mathbf{W}^{(k)}_{\mathcal{\Phi}}=\sum_{t=0}^{\infty}\mathbb{E} [\psi_{k}^{2t}\mathcal{\Phi}_{0:t-1}^{T}\mathcal{\Phi}_{0:t-1}]$.
\end{proposition}
\begin{sketchproof}
    $\widetilde{\mathbf{w}}^{(k)}_{t} = \sum_{\tau=0}^{t-1}\psi_{k}^{t-\tau-1}\mathcal{\Phi}_{\tau+1:t-1}(\psi_{k}\mathbf{S}_{\tau}-\mathbf{D})\mathbf{n}^{(k)}_{\tau}$. By total expectation \(\mathbb{E}[\mathbf{X}] = \mathbb{E}[\mathbb{E}[\mathbf{X}|\mathbf{Y}]]\), \(\textnormal{tr}(\mathbb{E}[\cdot]) = \mathbb{E}[\textnormal{tr}(\cdot)]\) and \(\text{tr}(\mathbf{AB}) = \textnormal{tr}(\mathbf{BA})\), $\mathbb{E}[\widetilde{\mathbf{w}}^{(k)}_{t}(\widetilde{\mathbf{w}}^{(k)}_{t})^{T}]=(\sigma^{(k)})^2\sum_{\tau=0}^{t-1}\psi_{k}^{2(t-\tau-1)}\mathbb{E}[$ $(\psi_{k}\mathbf{S}_{\tau}-\mathbf{D})^{T}\mathcal{\Phi}^{T}_{\tau+1:t-1}\mathcal{\Phi}_{\tau+1:t-1}(\psi_{k}\mathbf{S}_{\tau}-\mathbf{D})]$. Since \(\mathcal{\Phi}_{\tau+1:t-1}\) and \(\mathbf{S}_{\tau}\) are independent for each term in the summation, the recursive relation in Eq.~(12) implies that \(\mathbb{E}[\mathcal{\Phi}^{T}_{\tau+1:t-1} \mathcal{\Phi}_{\tau+1:t-1}]\) can be computed first. As \(t \to \infty\), Eq.~(14) follows.
\end{sketchproof}
\begin{lemma}
    If any graph realization $\mathcal{G}_{t}(\mathcal{V},\mathcal{E}_{t})$ satisfies $\Vert \psi_k \mathbf{S}_{t} \Vert_{2} \leq |\psi_k\rho|< 1$, then infinite series $\mathbf{W}^{(k)}_{\mathcal{\Phi}}=\sum_{t=0}^{\infty}\mathbb{E} [\psi_{k}^{2t}\mathcal{\Phi}_{0:t-1}^{T}\mathcal{\Phi}_{0:t-1}]$ from stochastic system converges, and satisfies 
\begin{equation}
\begin{aligned}
\mathbf{W}^{(k)}_{\mathcal{\Phi}} &= \psi_{k}^{2}\mathbb{E}[\mathbf{S}_{t}\mathbf{W}^{(k)}_{\mathcal{\Phi}} \mathbf{S}_{t}]+\mathbf{I}
\end{aligned}
\end{equation}
\end{lemma}
\begin{proof}
    See Appendix I.
\end{proof}
\begin{theorem}
    Under \textbf{Proposition 4} and \textbf{Lemma 1}, given $\mathbf{D}$ varies as $\mathbf{D}^{(k)}=\textnormal{diag}\{\alpha_{1}^{(k)},\cdots, \alpha_{N}^{(k)} \}$, $\Theta = \{\mathbf{D}^{(1)}\mathbf{1},\cdots, \mathbf{D}^{(K)}\mathbf{1}\}\in\mathbb{R}^{N\times K}$, the closed-form solution $\alpha_{i}^{(k)}$ and noise reduction $I(\Theta)$ are 
    \begin{equation}
\begin{aligned}
\alpha_{i}^{(k)} & = \frac{\psi_{k}[\mathbf{W}_{\mathcal{\Phi}}^{(k)}\overline{\mathbf{S}}]_{ii}}{[\mathbf{W}_{\mathcal{\Phi}}^{(k)}]_{ii}} \\
I(\Theta) & ==\sum_{k=1}^{K}\frac{(\sigma^{(k)})^2}{N}\sum_{i=1}^{N}\frac{\psi_{k}^{2}[\mathbf{W}_{\mathcal{\Phi}}^{(k)}\overline{\mathbf{S}}]_{ii}^{2}}{[\mathbf{W}_{\mathcal{\Phi}}^{(k)}]_{ii}}
\end{aligned}
\end{equation}
\end{theorem}
\begin{proof}
    The argument follows similarly to that of \textbf{Theorem 1}.
\end{proof}

\section{Numerical Experiments}
\label{sec:typestyle}

This section presents our simulation results, illustrated in Fig.~\ref{fig:res}. We consider a low-pass graph filtering task on the David Sensor Network \cite{18}, which has \(N = 64\) nodes and \(M = 236\) edges. FIR and ARMA graph filters are designed to approximate an ideal low-pass filter, following the benchmark in \cite{8}. The desired frequency response is defined as \(h(\lambda) = 1\) for \(\lambda < \lambda_c\) and \(0\) otherwise, where \(\mathbf{S} = \lambda_{\text{max}}^{-1} \mathbf{L}\) and \(\lambda_c = 0.5\). For the FIR filter, the filter order varies from 6 to 12 to construct \(h(\lambda)\). For the ARMA filter, we use ARMA\(_1\) defined as \(\mathbf{y} = (\mathbf{I} + 0.5 \mathbf{S})^{-1} \mathbf{x}\). In stochastic cases, the settings remain unchanged, but with random \(\mathbf{S}_{t}\) of edge sampling probability $p\in[0.1,1]$. Each case is evaluated over 10,000 realizations. A dither quantizer \cite{19,20} with a range of \([-1, 1]\) is used to maintain independence, and the input \(\mathbf{x}\) has a uniform spectrum, properly normalized to
prevent overflow.
The signal-to-noise ratio (SNR) for evaluating the performance of both quantization error feedback and non-feedback filtering processes is defined as:
\begin{equation}
\text{SNR} = \left( \frac{\Vert \mathbf{y}_{q} - \mathbf{y} \Vert^{2}_{2}}{\Vert \mathbf{y} \Vert^{2}_{2}} \right)^{-1}, \mathbf{y} = 
\begin{cases} 
\mathbf{y}, & \text{(unbiased)} \\ 
\mathbb{E}[\mathbf{y}], & \text{(biased)} 
\end{cases}
\end{equation}
where \(\mathbf{y}_q\) is quantization noise injected output, and \(\mathbf{y}\) is the original output without quantization. In this context, the presence or absence of bias indicates whether the SNR accounts for graph randomness, which introduces losses unrelated to quantization.

As shown in Fig.~\ref{fig:res}, our method consistently outperforms the nonfeedback quantization approach, with maximum improvements of 7.53 dB, 8.21 dB, 7.3 dB, and 1.81 dB (biased). Fig.~\ref{fig:res}(a) demonstrates that our method is applicable to various FIR graph filters with different orders, with the noise reduction being dependent on their filter responses. Fig.~\ref{fig:res}(b) shows that when ARMA$_{1}$ reaches a steady state, the feedback model consistently maintains an SNR gap compared to the non-feedback model. Fig.~\ref{fig:res}(c) illustrates that, considering only quantization noise, its energy is negatively correlated with link loss. This indicates that with sparser connections, the quantization error is less likely to propagate and accumulate during the fusion process. Finally, Fig.~\ref{fig:res}(d) shows that even when accounting for the bias introduced by link loss, the overall deviation of the proposed method remains lower than that of the non-feedback model, with its SNR approaching the upper bound, defined as the ratio of the power of the filtered mean $\mathbb{E}[\mathbf{y}]$, to its variance $Var(\mathbf{y})$.

\section{Conclusion}
\label{sec:majhead}

In conclusion, this work introduces an error feedback approach to reduce quantization noise in distributed graph filters. Theoretical analysis and closed-form error weights for both FIR and ARMA graph filters are derived on deterministic and random graph topologies. Numerical experiments confirm the method's effectiveness in noise reduction.

\begin{figure}[tb]
\begin{minipage}[b]{.49\linewidth}
  \centering
  \centerline{\includegraphics[width=4.05cm]{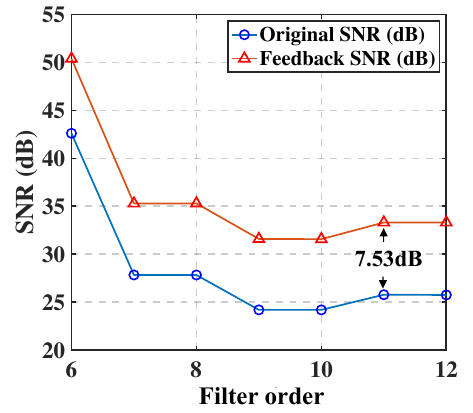}}
  \centerline{(a) Deterministic FIR}\medskip
\end{minipage}
\hfill
\begin{minipage}[b]{.49\linewidth}
  \centering
  \centerline{\includegraphics[width=4.07cm]{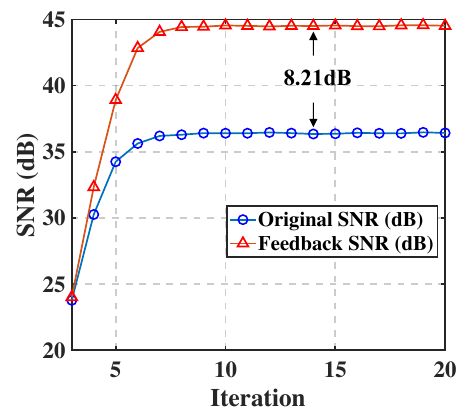}}
  \centerline{(b) Deterministic ARMA}\medskip
\end{minipage}
\hfill
\begin{minipage}[b]{0.49\linewidth}
  \centering
  \centerline{\includegraphics[width=4.05cm]{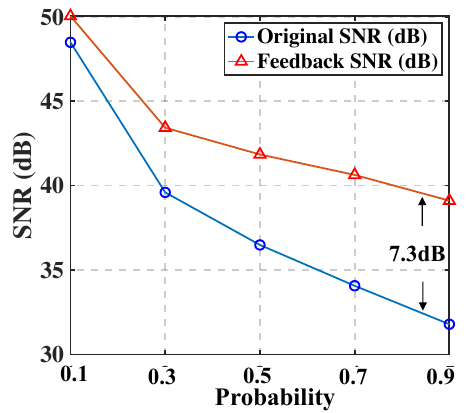}}
  \centerline{(c) Stochastic FIR}\medskip
\end{minipage}
\hfill
\begin{minipage}[b]{0.49\linewidth}
  \centering
  \centerline{\includegraphics[width=4.07cm]{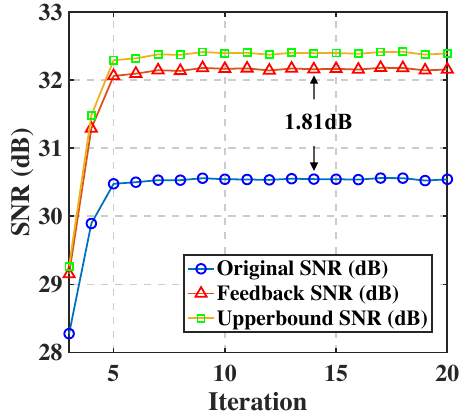}}
  \centerline{(d) Stochastic ARMA}\medskip
\end{minipage}
\vspace{-10pt} 
\caption{Performance comparison.(a) The quantization bits for these filters are $b = \{14\times 1, 16\times 2, 20\times 2, 25\times 2\}
$. (b) ARMA$_{1}$ with $b ={6}$. (c) FIR of order $K=6$ with $b ={12}$. (d) ARMA$_{1}$ with $b ={4}$ under $p=0.95$.}
\label{fig:res}
\end{figure}

\section{APPENDIX I}

\label{sec:app}

\textit{Proof of Theorem 1:} Let $\mathbf{D}_{k-1}=\text{diag}\{\alpha_{1, k-1},\cdots, \alpha_{N, k-1} \}$ replace $\mathbf{D}$ in Eq. (\ref{deqn_ex9}), we have mitigation term $I_{k-1}=$
\begin{equation}
\begin{aligned}
    -2\sum_{i=1}^{N}\Big[\mathbf{G}_{\mathbf{n}_{k-1}}(\mathbf{S}) \mathbf{S}\Big]_{ii} \alpha_{i, k-1}+ \sum_{i=1}^{N}\Big[\mathbf{G}_{\mathbf{n}_{k-1}}(\mathbf{S})\Big]_{ii}\alpha_{i, k-1}^{2}  
\end{aligned}
\label{deqn_ex11}
\end{equation}
Let the gradient of $I_{k-1}$ be zero, i.e.,
\begin{equation}
\begin{aligned}
\frac{\partial I_{k-1}}{\partial \alpha_{i, k-1}} =\Big[\mathbf{G}_{\mathbf{n}_{k-1}}(\mathbf{S})\Big]_{ii}\alpha_{i, k-1}-\Big[\mathbf{G}_{\mathbf{n}_{k-1}}(\mathbf{S}) \mathbf{S}\Big]_{ii} = 0
\end{aligned}
\label{deqn_ex12}
\end{equation}
we have
\begin{equation}
\begin{aligned}
\alpha_{i,k-1} = \frac{\Big [\mathbf{G}_{\mathbf{n}_{k-1}}(\mathbf{S})\mathbf{S}\Big ]_{ii}}{\Big [ \mathbf{G}_{\mathbf{n}_{k-1}}(\mathbf{S})\Big ]_{ii}}, I_{k-1} =\sum_{i=1}^{N}{\frac{\Big [\mathbf{G}_{\mathbf{n}_{k-1}}(\mathbf{S})\mathbf{S}\Big ]^{2}_{ii}}{\Big [ \mathbf{G}_{\mathbf{n}_{k-1}}(\mathbf{S})\Big ]_{ii}}}
\end{aligned}
\label{deqn_ex13}
\end{equation}
The proof is completed.

\vspace{1em}
\noindent \textit{Proof of Lemma 1:} By analogy to Carl Neumann’s theorem \cite{21},we consider the convergence of matrix norms. Applying Jensen's inequality, any subterm of \(\mathbf{W}^{(k)}_{\mathcal{\Phi}}\) satisfies
\begin{equation}
\begin{aligned}
\left\Vert \mathbb{E} \left[\psi_{k}^{2t} \mathcal{\Phi}_{0:t-1}^{T} \mathcal{\Phi}_{0:t-1}\right] \right\Vert_{2} \leq \mathbb{E} \left[\left\Vert \psi_{k}^{2t} \mathcal{\Phi}_{0:t-1}^{T} \mathcal{\Phi}_{0:t-1} \right\Vert_{2}\right].
\end{aligned}
\end{equation}
From sub-multiplicativity property of the spectral norm, \(\|\mathbf{AB}\|_{2} \leq \|\mathbf{A}\|_{2} \|\mathbf{B}\|_{2}\), we have
\begin{equation}
\begin{aligned}
\mathbb{E} \left[\left\Vert \psi_{k}^{2t} \mathcal{\Phi}_{0:t-1}^{T} \mathcal{\Phi}_{0:t-1} \right\Vert_{2}\right] \leq \mathbb{E} \left[(\|\psi_{k} \mathbf{S}_{t}\|_{2})^{2t}\right].
\end{aligned}
\end{equation}
Since \(\|\psi_k \mathbf{S}_{t}\|_{2} \leq |\psi_k \rho| < 1\), it follows that
\begin{equation}
\begin{aligned}
\left\Vert \mathbb{E} \left[\psi_{k}^{2t} \mathcal{\Phi}_{0:t-1}^{T} \mathcal{\Phi}_{0:t-1}\right] \right\Vert_{2} \leq \mathbb{E} \left[(\|\psi_{k} \mathbf{S}_{t}\|_{2})^{2t}\right] \leq |\psi_k \rho|^{2t}.
\end{aligned}
\end{equation}
Since \(0 < |\psi_k \rho| < 1\), the term \(|\psi_k \rho|^{2t}\) converges, then $\mathbf{W}^{(k)}_{\mathcal{\Phi}}  = \sum_{t=0}^{\infty} \mathbb{E} \left[\psi_{k}^{2t} \mathcal{\Phi}_{0:t-1}^{T} \mathcal{\Phi}_{0:t-1}\right]$ converges. By total expectation \(\mathbb{E}[\mathbf{X}] = \mathbb{E}[\mathbb{E}[\mathbf{X}|\mathbf{Y}]]\), Eq.~(15) follows. The proof is completed.

\vfill\pagebreak


\end{document}